\documentclass[12,reqno]{amsart} 
\usepackage{amsmath,amssymb,amsthm,mathtools}
\usepackage{float}\restylefloat{figure}
\usepackage[pdftex]{graphicx}
\usepackage[usenames]{color}
\usepackage{hyperref}
\usepackage[leftmargin=15pt,vskip=8pt]{quoting}
\usepackage{url}
\newtheorem{theorem}{Theorem}[section]
\newtheorem{corollary}[theorem]{Corollary}
\newtheorem{claim}[theorem]{Claim}
\newtheorem{lemma}[theorem]{Lemma}

\newtheorem{proposition}[theorem]{Proposition}

\theoremstyle{definition}

\newtheorem{definition}[theorem]{Definition}
\newtheorem{example}{Example}

\newtheorem{remark}[theorem]{Remark}

\newcommand\ang[1]{\ensuremath{\big\langle#1\big\rangle}}

\newcommand{\bE}{\ensuremath{\textbf{E}}}
\newcommand{\bra}[1]{\ensuremath{\langle#1|}}
\newcommand{\braket}[2]{\ensuremath{\langle#1\,|\,#2\rangle}}
\newcommand\C{\ensuremath{\mathbb C}}
\newcommand\Dom[1]{\ensuremath{\textrm{Dom}{(#1)}}}

\newcommand\E{\ensuremath{\mathcal E}}
\newcommand\ES{\ensuremath{\mathrm{ES}}}
\newcommand\EV{\ensuremath{\mathrm{EV}}}

\renewcommand{\H}{\ensuremath{\mathcal H}}
\newcommand{\ket}[1]{\ensuremath{|#1\rangle}}

\renewcommand\O{\ensuremath{\mathcal O}}

\newcommand\ox{\ensuremath{\otimes}}
\renewcommand\P{\ensuremath{\mathcal P}}
\renewcommand\phi{\varphi}
\newcommand\qef{\hfill$\triangleleft$} 
\newcommand\qefhere{\tag*{$\triangleleft$}} 
\newcommand\R{\ensuremath{\mathbb R}}
\newcommand\cS{\ensuremath{\mathcal S}}

\newcommand{\notarrow}{\kern .42em\not\kern -.42em\longrightarrow}

\newenvironment{sm}
  {\left(\begin{smallmatrix}}
  {\end{smallmatrix}\right)}

\title[Negative probabilities]{Negative probabilities: \\
  What are they for?}
\author{Andreas Blass}
\address{Mathematics Department\\
University of Michigan\\
Ann Arbor, MI 48109--1043, U.S.A.}
\email{ablass@umich.edu}
\author{Yuri Gurevich}
\address{Computer Science and Engineering\\
University of Michigan\\
Ann Arbor, MI  48109-2121, U.S.A}
\email{gurevich@umich.edu}
\thanks{Partially supported by the US Army Research Office under W911NF-20-1-0297}

\begin{document}

\begin{abstract}
An \emph{observation   space} \cS\ is a family of probability distributions \ang{P_i: i\in I} sharing a common sample space $\Omega$ in a consistent way.
A \emph{grounding} for $\cS$ is a signed probability distribution $\P$ on $\Omega$ yielding the correct marginal distribution $P_i$ for every $i$.
A wide variety of quantum scenarios can be formalized as observation spaces.
We describe all groundings for a number of quantum observation spaces.
Our main technical result is a rigorous proof that Wigner's distribution is the unique signed probability distribution yielding the correct marginal distributions for position and momentum and all their linear combinations.
\end{abstract}
\maketitle

\begin{quote}\raggedleft\small\it
What are numbers and what are they for?\\
--- Richard Dedekind \cite{Dedekind}
\end{quote}

\section{Introduction}
\label{intro}

The uncertainty principle asserts a limit to the precision with which position $x$ and momentum $p$ of a particle can be known simultaneously.
You may know the probability distributions of $x$ and $p$ individually in a given quantum state but the joint probability distribution with these marginal distributions of $x$ and $p$ makes no physical sense.
Yet the question whether there exists an appropriate joint distribution (with the correct marginal distributions of $x$ and $p$) makes mathematical sense.  In 1932, Eugene Wigner exhibited such a joint distribution \cite{Wigner}. Some of its values were negative, but this, wrote Wigner, ``must not hinder the use of it in calculations as an auxiliary function which obeys many relations we would expect from such a probability.''

But what are negative probabilities? We split this intriguing question into two: formal and ontological.
The formal question asks how to generalize the axiomatic probability theory of Kolmogorov so that negative probabilities are allowed.
Kolmogorov's theory is based on measure theory, and measure theory provides the desired generalization: Signed probability distributions are special signed measures. In \S\ref{sec:prob}, we define a \emph{signed probability distribution} on a measurable space as a countably additive function assigning a real number to every measurable set and assigning 1 to the set of all points.

The ontological question is much harder. It asks what reality, or at least  intuition, is behind negative probabilities. Obviously, the standard frequential interpretation of probabilities does not apply to negative probabilities.
There are serious attempts to address the ontological question \cite{Feynman,AB} but, at least in our judgement, the ontological question remains wide open.

The question in the center of our attention here is more pragmatic: What are negative probabilities good for? It is not rare in science to usefully apply a notion without properly understanding its ontology. One historical example is complex numbers. They are not quantities, but they could be used to solve algebraic equations. Another historical example is the use of uncountable sets in mathematical analysis.

There are efforts to apply negative probabilities in several disciplines, including biology \cite{Hakulinen}, decision theory \cite{de Barros}, finance \cite{Meissner}, information theory \cite{Liu}, machine learning \cite{Zhao}, and transportation science \cite{Xie}. But the vast majority of negative-probability applications are in quantum physics. Wigner's discovery, mentioned above, led to a whole new approach to foundations of quantum mechanics \cite{Zachos} and to fruitful physical applications, especially in quantum tomography; see \cite{Smithey, UIUC,W-QT} for example.

Many of the disparate quantum applications can be seen as examples of a certain application template. Thinking of Wigner's 1932 discovery provoked us to introduce the template.

According to quantum mechanics, some but not all observables can be measured simultaneously, and the theory predicts probabilities for the outcomes of such measurements. Thus, for a given physical system, we may have numerous coexisting probability distributions which are impossible to subsume under a single distribution for all the observables. But negative probabilities may allow us the desired single distribution, which we call a grounding and which may be useful for analysis, computation and even applications.

In more detail, consider a family, indexed by an arbitrary set $I$, of probability distributions \ang{\P_i: i\in I} sharing a common sample space $\Omega$ in a consistent way, so that any two of the distributions agree on the events where both of them are defined.
We call members of the domain \Dom{P_i} of any $P_i$ \emph{observable events}, and we call observable events \ang{e_j: j\in J} \emph{coobservable} if they all belong to \Dom{P_i} for the same $i$. Finally, we call the pair $\cS = (\Omega,\ang{\P_i: i\in I})$ an \emph{observation space}.

The idea behind observability is this. Each $P_i$ reflects a probabilistic experiment, e.g.\ a quantum measurement, and the information obtained by running the experiment specifies an atom $a$ of the Boolean algebra \Dom{P_i}.
While atom $a$ is nonempty, it need not be a singleton (as we will see below).
All events $e\in\Dom{P_i}$ with $a\subseteq e$ occurred in this run of the experiment, and the rest of the events in \Dom{P_i}, those disjoint from $a$, did not occur in this run. In that sense, we observe, for any run and any $e\in\Dom{P_i}$, whether $e$ occurs or not.

It is reasonable to ask whether the legitimate combinations of
observable events, those in the $\sigma$-algebra $\Sigma$ generated by the observable events, can be assigned probabilities in a consistent way.
Our notion of grounding formalizes the hope of doing just that. A \emph{grounding}, or \emph{ground distribution}, for the observation space \cS\ is a signed probability distribution $\P$ on the measurable space $(\Omega,\Sigma)$ such that the original distributions $\P_i$ are restrictions of $\P$.
The \emph{grounding problem} for  \cS\ is the problem of describing  the groundings for \cS. Is there a grounding? Is there a grounding which is nonnegative, i.e., has no negative values? Is the grounding unique? What are the general forms of groundings and of nonnegative groundings?

In \S\ref{sec:obs}, we study observation spaces and prove a modeling theorem according to which a wide variety of quantum and classical scenarios can be formalized as observation spaces.

In many cases, the resulting observation space is finite.
\S\ref{sec:finite} is devoted to quantum scenarios modeled by finite observation spaces. We solve the grounding problem for a number of observation spaces.
In particular, we address the scenario, studied by Feynman, involving a spin-$\frac12$ particle and two components of its spin \cite{Feynman}.
The uncertainty principle implies that these two quantities cannot have definite values simultaneously.
So it seems plausible that an attempt to assign joint
probabilities would again, as in Wigner's case, lead necessarily to
negative probabilities.
It turns out that negative probabilities can be avoided in this situation. Specifically, the groundings for Feynman's scenario form an infinite family with one real parameter $t$, and there is a nonempty real interval $[m,M]$ such that the grounding is nonnegative if and only if $m\le t\le M$.
We also consider a Feynman-type scenario involving three independent components of a particle's spin. There are nonnegative groundings in that scenario as well. We give a general solution for all groundings and for all nonnegative groundings in that scenario.

The proof of the modeling theorem gives us a general method of
modeling quantum scenarios by observation spaces. But, in specific
scenarios, e.g.\ in all scenarios studied in \S\ref{sec:finite}, modeling can be done more economically and more faithfully.

In \S\ref{sec:context}, we show that, in contextual situations like that in the Kochen-Specker theorem, such faithful modeling is impossible.

Wigner's distribution happens to be the unique signed probability distribution on the phase space that yields the correct marginal distributions not only for position and momentum but for all their linear combinations \cite{Bertrand}. In other words, it is the unique grounding for the observation space formed by the distributions of these linear combinations. In \S\ref{sec:w}, we give a rigorous proof of that fact.

\subsection*{Acknowledgement.}

We thank Alexander Volberg for contributing Lemma~\ref{lem:sasha} and Vladimir Vovk for a useful correction.

\section{Signed probability distributions}
\label{sec:prob}

A \emph{measurable space} is a set $\Omega$ together with a
$\sigma$-algebra of subsets of $\Omega$. Elements of $\Omega$ are
\emph{sample points}, $\Omega$ is the \emph{sample space} and members
of the $\sigma$-algebra are \emph{measurable sets}. In this paper, by
default, the sample space is not empty.

\begin{definition}\label{def:sp}
  A \emph{signed probability distribution} on a measurable space
  $(\Omega,\Sigma)$ is a real-valued, countably additive function on $\Sigma$ assigning value 1 to $\Omega$. If $e\in \Dom{\P}$, then the number $\P(e)$ is the \emph{probability} of $e$. If all
  probabilities are nonnegative, then $\P$ is \emph{nonnegative}. \qef
\end{definition}

One may worry whether countable additivity makes sense in  signed
probability spaces, i.e., whether $\P(\bigcup_n e_n) = \sum_n
\P(e_n)$ whenever events $e_n$ are pairwise disjoint. A priori, the
sum could depend on the order of the events. But in fact it does not.
Indeed, $\sum \{\P(e_n): \P(e_n)\ge0\}$ converges to the probability $p$ of event
$\bigcup\{e_n: \P(e_n)\ge0\}$, and  $\sum \{\P(e_n): \P(e_n)<0\}$
converges to the probability $q$ of event $\bigcup\{e_n: \P(e_n)<0\}$. Accordingly, $\sum_n \P(e_n)$ converges absolutely to $p + q$, and therefore the order of summands is irrelevant.

Many laws of nonnegative probability theory survive the generalization to signed probability theory.
\begin{quoting}
``Since the formal structure for the distributions is unaltered, it is trivial to show that the probabilities in the extended theory must for consistency obey the same rules as before, that is, the addition and multiplication laws." \cite[p.~72]{Bartlett}
\end{quoting}
For example, we can use, in the signed case, the standard definitions of random variable, expectation, independence, and standard deviation, and we still have
\begin{align*}
\bE(aX + bY) &\ =\  a\bE(X) + b\bE(Y),\\
X,Y\textrm{ are independent } &\implies \bE(XY) = \bE(X)\bE(Y),\\
\sigma(aX+b) &= |a|\cdot \sigma(X).
\end{align*}
But one should be careful. For example, the following laws fail in the signed case.
\begin{align*}
 &\phantom{\ \implies}\P(e) \le 1,\\
\big(\P(e) = 0\textrm{ and } \omega\in e\big) &\implies \P(\omega) = 0,\\
 \big( X(\omega)\le Y(\omega)\textrm{ for all sample points }\omega
  \big) &\implies \bE(X)\le\bE(Y).
\end{align*}

Probabilistic experiments, whether real-world experiments or  thought experiments, may give rise to a signed probability distribution.

\begin{example}[Piponi's thought experiment \cite{Piponi}]
A machine produces boxes with pairs $(l,r)$ of bits.
In any run of the experiment, exactly one of the following three tests can be done.
\begin{enumerate}
\item Look through the left window and observe $l$.
\item Look through the right window and observe $r$
\item Test whether $l=r$.
\end{enumerate}
Performing these tests repeatedly, you find out that
\begin{itemize}
\item in the left window you always see 1,
\item in the right window you always see 1,
\item the two bits are always different.
\end{itemize}
This gives rise to a signed probability distribution $\P$ on the
sample space $\Omega = \big\{(l,r): l,r\in\{0,1\}\big\}$.
Let $p_{lr}
= \P(l,r)$, so that
\[ p_{00} + p_{01} + p_{10} + p_{11} = 1 .\]
The three tests give three additional constraints:
\[ p_{10} + p_{11} = 1, \quad
   p_{01} + p_{11} = 1, \quad
   p_{01} + p_{10} = 1.
\]
The unique solution of this system of four linear equations is this:
\[ p_{00} = -\frac12,\quad p_{01} = p_{10} = p_{11} = \frac12.
\qefhere \]
\end{example}

Piponi's thought experiment allows one to perform some computations, e.g.,
\begin{align*}
\bE(l) &= 0\cdot(-1/2) + 0\cdot1/2 + 1\cdot1/2 + 1\cdot1/2\\
&= (1/2)[-0 + 0 + 1 + 1] = 1 = \bE(r),\\
\bE(l+r) &= (1/2)[-0 + 1 + 1 + 2] = 2,\\
\sigma(l) &= \sqrt{\bE[(l-1)^2]}
= \sqrt{(1/2)[-1 + 1 + 0 + 0]} = 0 = \sigma(r),\\
\sigma(l+r) &= \sqrt{\bE[(l+r-2)^2]}
= \sqrt{(1/2)[-4 + 1 + 1 + 0]} = \sqrt{-1}.
\end{align*}
Note some unusual features of standard deviation. The most glaring is
the imaginary standard deviation of $l+r$. But there is something
else. In nonnegative probability theory, a  standard deviation of zero
indicates that the random variable is almost everywhere constant. This
still works for $l$ in the sense that $l=1$ on the set $\{10,11\}$
whose complement $\{00,01\}$ has measure $-1/2 +1/2 = 0$. Similarly
$r$ is almost everywhere constant. But it is not the case that $l+r$ is almost everywhere constant; in signed probability
spaces, the union of sets of measure zero need not have measure zero.

Does Piponi's thought experiment make any physical sense?
Maybe. Perhaps one can come up with an appropriate real-world quantum
experiment involving two qubits and three projective yes/no
measurements.

But physics isn't the only prospective application domain for signed probability distributions.
Piponi's experiment may make some sense in social sciences. It may be
possible that
\begin{itemize}
\item one party provides overwhelming evidence in favor of a claim $L$,
\item the other party provides overwhelming evidence in favor of an
  alternative claim $R$, and
\item in all circumstances, exactly one of the two claims is true and
  the other is false; sometimes $L$ is true and sometimes $R$ is
  true.
\end{itemize}

\section{Observation spaces}
\label{sec:obs}

We introduce a general framework of observation spaces for the sort of situation that arises in Piponi's experiment.
Certain events and
combinations of events are observable and have probabilities. Other
combinations of events cannot be simultaneously observed and need not
have well defined probabilities. It is reasonable to ask whether such
combinations can be assigned probabilities in a way consistent with
the given probabilities of observable events.
Our notion of observation space will formalize the picture of simultaneously observable events and their probability distributions.
Our notion of grounding will formalize the hope of consistently assigning probabilities even to combinations of events which are not simultaneously observable.

\begin{definition}\label{def:OS}
An \emph{observation space} is a nonempty set $\Omega$ together with nonnegative probability distributions \ang{\P_i: i\in I} with $\sigma$-algebras $\Dom{\P_i}$, subject to the \emph{coherence requirement}
\[ e\in\Dom{\P_i}\cap\Dom{\P_j} \implies \P_i(e) = \P_j(e). \qefhere\]
\end{definition}

Elements of $\Omega$ are \emph{sample points}, subsets of $\Omega$ are \emph{events}, and $\Omega$ itself is the \emph{underlying sample space}.
An event $e$ is \emph{observable} if it belongs to some
$\Dom{\P_i}$. A family $E$ of events is \emph{coobservable} if $E\subseteq \Dom{\P_i}$ for some $i\in I$.

Piponi's thought experiment gives rise to an observation space with
sample space $\{00,01,10,11\}$ and three nonnegative probability
distributions $\P_1, \P_2, \P_3$ corresponding to the tests
(1)--(3). In addition to $\emptyset$ and $\{00,01,10,11\}$, the
Boolean algebras $\Dom{\P_1}$, $\Dom{\P_2}$ and $\Dom{\P_3}$  contain complementary atoms
\begin{align*}
(1)\quad & \{00,01\}, \{10,11\},\\
(2)\quad & \{00,10\}, \{01,11\},\\
(3)\quad & \{00,11\}, \{01,10\}
\end{align*}
respectively. \Dom{\P_1} consists of the events that can be observed
by looking into the left window. \Dom{\P_2} and \Dom{\P_3} represent
looking into the right window and testing the equality
respectively. An individual outcome, like 01, represents an unobservable combination of events, namely, looking into the left and the right windows simultaneously (or looking into one of the windows and testing equality simultaneously).

\begin{definition}
A \emph{ground distribution}, or \emph{grounding}, of an observation space $\left(\Omega,\ang{\P_i: i\in I}\right)$ is a signed probability distribution $\P$ on $\Omega$ such that
\begin{itemize}
\item \Dom{\P} is the $\sigma$-algebra generated by all of the observable events and
\item every $\P_i$ is the restriction of $\P$ to $\Dom{\P_i}$. \qef
\end{itemize}
\end{definition}

The  \emph{grounding problem} for a given observation space is the
problem of characterizing the groundings for the
observation space:
\begin{itemize}
\item Is there a grounding? Is there a nonnegative grounding?
\item Is the grounding unique?
\item What are the general forms of signed groundings and of
  nonnegative groundings?
\end{itemize}

The grounding problem is trivial if the given observation space has
only one original distribution.
In this case, the original distribution is the unique grounding. But the case of two original distributions is already nontrivial as we will see in \S\ref{sub:f1}.

A wide variety of quantum (and also classical) scenarios can be formalized as observation spaces. But, in this section, by default, we work with a fixed finite-dimensional Hilbert space \H, so that the spectra of observables are pure point spectra%
\footnote{An alternative restriction, sufficient for our purposes in this section, is that we restrict attention to observables with pure point spectra.}.

\begin{definition}              \label{def:mte}
A \emph{multi-test experiment} over \H\ is a pair
$\big( \ket\psi, \O \big)$ where
\ket{\psi} is a unit vector in \H, and
\O\ is a set of self-adjoint operators, called \emph{observables}, on \H, subject to the following requirement.
\begin{description}
\item[\texttt{Simultaneous measurement}]
If observables $A_1, \dots, A_k\in\O$ commute, then there is an observable $B\in\O$ such that every eigenspace of $B$ is an intersection $X_1\cap\ \dots\ \cap X_k$ of eigenspaces of $A_1, \dots, A_k$ respectively%
\footnote{The requirement follows from its $k=2$ version.}. \qef
\end{description}
\end{definition}

Intuitively, a multi-test experiment is a laboratory experiment or thought experiment where a source repeatedly produces quantum systems in the specified state \ket{\psi}, and each time an observable $A$, arbitrarily chosen from \O, is measured on the system. We call such a measurement a test; hence the name ``multi-test experiment."

If the observables $A_1, A_2\in \O$ commute, then, according to quantum mechanics, we can measure them simultaneously in \ket\psi.
The simultaneous-measurement condition provides an observable $B$ such that measuring $B$ amounts to measuring $A_1$ and $A_2$ together. In that sense, the condition seems rather natural to us.
In particular, the condition is trivially satisfied if \O\ contains no commuting observables, as will be the case in all four experiments in \S\ref{sec:finite}.

\begin{lemma}\label{lem:sm}
Let \O\ be a set of observables satisfying the simultaneous-measurement requirement. If observables $A_1, A_2, \dots, A_k\in\O$ commute, then there is an observable $B$ in \O\ such that every eigenspace of $B$ has the form $X_1\cap X_2\cap \cdots\cap X_k$ where $X_1, X_2, \dots, X_k$ are eigenspaces of $A_1, A_2, \dots, A_k$ respectively.
\end{lemma}

\begin{proof}
Induction on $k$ with the trivial case $k=1$ being the base of induction.
Suppose that the claim is proven for $k$, and let $A_1,\dots, A_{k+1}$ be commuting observables in \O.
By the induction hypothesis, there is an observable $B\in\O$ such that every eigenspace of $B$ has the form $X_1\cap \cdots\cap X_k$ where $X_1,\dots, X_k$ are eigenspaces of $A_1, \dots, A_k$ respectively. It follows that $B$ commutes with $A_{k+1}$.
By the simultaneous-measurement requirement, there is an observable $C\in\O$ such that every eigenspace of $C$ has the form $Y\cap X_{k+1}$ where $Y,X_{k+1}$ are eigenspaces of $B, A_{k+1}$ respectively. Therefore every eigenspace of $C$ has the form $X_1\cap \cdots\cap X_k\cap X_{k+1}$ where $X_1, \dots, X_k, X_{k+1}$ are eigenspaces of $A_1,\dots, A_k, A_{k+1}$ respectively.
\end{proof}

If $A$ is an observable, let $\EV(A)$ be the collection of the eigenvalues of $A$, i.e., the spectrum of $A$.
For each $r\in\EV(A)$, let $E_r(A)$ be the (maximal) eigenspace of $A$ for eigenvalue $r$, and let $\ES(A) = \{E_r(A): r\in\EV(A)\}$.

\begin{definition}\label{def:model}
An observation space $\Big(\Omega,\ang{\P_A: A\in\O}\Big)$ \emph{models} a multi-test experiment $\E = \big(\ket\psi,\O\big)$ over \H\ if, for each $A\in\O$, there is a map $\mu_A: \ES(A) \to 2^\Omega$ such that
\begin{description}
\item[\tt Partition] for every $A\in\O$, events $\alpha_r(A) = \mu_A(E_r(A))$, where $r\in\EV(A)$, partition $\Omega$, and
\item[\tt Correctness] $\P_A(\alpha_r(A))$ is the probability that, according to quantum mechanics, the measurement of $A$ in state \ket{\psi} exhibits $r$. \qef
\end{description}
\end{definition}

\begin{theorem}\label{thm:model}
Every multi-test experiment is modeled by some observation space.
\end{theorem}

\begin{proof}
Consider a multi-test experiment $\E = \big( \ket{\psi}, \O \big)$.
For every $A\in\O$ and every $r\in\EV(A)$, let $p_r(A)$ be the
probability that, according to quantum mechanics, the measurement of $A$ in state \ket{\psi} exhibits $r$.

We construct an observation space $\cS_0 = \big(\Omega,\ang{\P_A: A\in \O}\big)$ modeling \E. Define
\[ \Omega = \prod \ang{\EV(A): A\in\O}. \]
For every $A\in\O$, let
\[ \alpha_r(A) = \{f\in\Omega: f(A) = r\}
\quad \textrm{for every }r\in\EV(A),\]
so that the events \ang{\alpha_r(A): r\in\EV(A)} partition $\Omega$. Define $\P_A$ so that the domain of $\P_A$ is the Boolean algebra generated by the  atoms $\alpha_r(A)$, and each $\P_A (\alpha_r(A)) = p_r(A)$.

The required coherence of $\cS_0$ follows from the claim that, for all indexes $A\ne B$ in \O, the intersection
$\Dom{\P_A} \cap \Dom{\P_B}$
contains only the empty set and the whole $\Omega$.
To prove the claim, let $e$ be a nonempty event in $\Dom{\P_A} \cap \Dom{\P_B}$.

There are nonempty subsets $X,Y$ of $\EV(A), \EV(B)$ respectively such that
\[e = \bigcup_{r\in X} \alpha_r(A)
    = \bigcup_{s\in Y} \alpha_s(B).\]
Choose any $r\in X$ and any
$f\in \alpha_r(A) = \{g\in \Omega: g(A) = r\}$.
Change the $B^{th}$ component of $f$ to any $s\in \EV(B)$;
the resulting element of $\Omega$ still belongs to $\alpha_r(A)\subseteq e$.
It follows that $Y = \EV(B)$ and therefore $e = \Omega$.

This completes the proof of the theorem.
\end{proof}

The observation space $\cS_0$ constructed in the proof of Theorem~\ref{thm:model} models the given multi-test experiment \E\ in a perfunctory way.
In particular, $\cS_0$ always has a nonnegative grounding; just use the product measure on $\Omega$, i.e., make the results of tests probabilistically independent.
Still, the theorem demonstrates the existence of some observation space for \E.
One does not have to be a proponent of hidden variables to realize that a grounding of an observation space for \E\ can be useful in the mathematical analysis of the experiment.

\section{Grounding some observation spaces}
\label{sec:finite}

Given a multi-test experiment \E, it is often possible to construct and analyze an observation space for \E\ which is more succinct than the observation space $\cS_0$ constructed in the proof of Theorem~\ref{thm:model} and which models \E\ more faithfully.

We  illustrate this point on four examples in subsections \S\ref{sub:f1} -- \ref{sub:hardy}.
The example in \S\ref{sub:f2} is new; the other three
examples come from the literature.
Notice also that the four-point observation space for Piponi's example in \S\ref{sec:obs} is also more succinct than the corresponding $\cS_0$, which would have eight sample points.

Intuitively, the proof of Theorem~\ref{thm:model} takes unfair advantage of the absence, in Definition~\ref{def:model} of modeling, of any correlation between the maps $\mu_A$ for different observables $A$.
A more faithful notion of modeling is provided by the following definition, which correlates the various maps $\mu_A$ in a meaningful way.

\begin{definition}              \label{def:monotone}
  An observation space $\Big(\Omega,\ang{\P_A: A\in\O}\Big)$
  \emph{monotonically models} a multi-test experiment
  $\E = \big(\ket\psi,\O\big)$ over \H\ if it models this experiment
  via maps $\mu_A: \ES(A) \to 2^\Omega$ as in
  Definition~\ref{def:model} and satisfies the following additional requirement.
  \begin{description}
  \item[\tt Monotonicity] If $A,B\in\O$ commute, $X\in\ES(A)$, $Y\in\ES(B)$, and $X\subseteq Y$, then $\mu_A(X)\subseteq\mu_B(Y)$. \qef
  \end{description}
\end{definition}

The idea behind this definition is this. As in the monotonicity requirement, suppose that $A,B\in\O$ commute, and suppose $X = E_r(A) \subseteq E_s(B) = Y$. Then observing the value $r$ for $A$ entails observing the value $s$ for $B$. So a sample point giving $A$ the value $r$ should also give $B$ the value $s$. This is one aspect of what we mean by ``faithful modeling.''

The monotonicity requirement is trivially satisfied if \O\ contains no commuting observables, as will be the case in all four experiments considered below in this section.
But we shall show in Section~\ref{sec:context} that some multi-test experiments cannot be monotonically modeled by observation spaces.

\subsection{Feynman's experiment}
\label{sub:f1}

Let $\H = \C^2$ and recall the Pauli operators given by matrices
\[
X=
\begin{pmatrix}
  0&1\\1&0
\end{pmatrix},\qquad Y=
\begin{pmatrix}
  0&-i\\i&0
\end{pmatrix},\qquad Z=
\begin{pmatrix}
  1&0\\0&-1
\end{pmatrix}.
\]
in the computational basis \ket0, \ket1 of \H.
Each of the Pauli operators $P$ has eigenvalues $\pm1$. The
eigenvectors for $+1$ and $-1$ of $P$ are the eigenvectors for
eigenvalues $1$ and $0$ of operator $(I+P)/2$ and also eigenvectors
for eigenvalues $0$ and $1$ respectively of operator $(I-P)/2$; here
$I$ is, as usual, the identity operator.
If you measure $P$ in state \ket{\psi}, the probability of obtaining
$+1$ in \ket{\psi} is $\left(1+\ang P\right)/2$ and the probability of
obtaining $-1$ in \ket{\psi} is $\left(1-\ang P\right)/2$; here \ang P
is, as usual, the expectation of $P$ in state \ket{\psi}.

For any unit vector \ket{\psi} in \H, consider a multi-test experiment
with two tests in state \ket{\psi}: measuring $Z$ and measuring $X$;
this is essentially the experiment studied by Richard Feynman in
\cite{Feynman}.

The experiment gives rise to an observation space \cS\ with sample
space $\Omega$ and two nonnegative probability distributions $\P_Z$
and $\P_X$ described below.

Since $Z$ and $X$ have eigenvalues $\pm1$, each run of the $Z$ or $X$
test produces $+1$ or $-1$, and the result may be described by $+$ or
$-$ respectively. Accordingly, the sample space $\Omega = \{++, +-,
-+, --\}$ where the first sign in each pair refers to $Z$ and the
second to $X$. The domain of $\P_Z$ is the Boolean algebra of events
generated by events
\[ +* = \{++,+-\}\textrm{\quad and\quad } -* = \{-+,--\}, \]
and the domain of $\P_X$ is the Boolean algebra of events generated by
events
\[*+ = \{++,-+\}\textrm{\quad and\quad } *- = \{+-,--\}.\]

In a given state \ket{\psi}, $\P_Z(+*) = \frac12(1+\ang Z)$ and
$\P_Z(-*) = \frac12(1-\ang Z)$, while $\P_X(*+) = \frac12(1+\ang X)$
and $\P_X(*-) = \frac12(1-\ang X)$.

\medskip\noindent{\tt Grounding problem.} Let $\P$ be an alleged
signed ground probability distribution for observation space
\cS. Using Feynman's notation, we abbreviate $\P(++)$ to $f_{++}$ and
similarly for the other three sample points $+-, -+$ and $--$. We have
\begin{align*}
f_{++}+f_{+-}&= \P_Z(+*) = \frac12(1+\ang Z),\\
f_{-+}+f_{--}&= \P_Z(-*) = \frac12(1-\ang Z),\\
f_{++}+f_{-+}&= \P_X(*+) = \frac12(1+\ang X),\\
f_{+-}+f_{--}&= \P_X(*-) = \frac12(1-\ang X).
\end{align*}
There is a redundancy in the equations. For example, the fourth
equation is obtained by adding the first two and subtracting the
third. But the first three are independent, so there's one free
parameter in the general solution.  In fact, it's easy to write down
the general solution:
\begin{align*}
f_{++}&=\frac14(1+\ang Z+\ang X+t)\\
f_{+-}&=\frac14(1+\ang Z-\ang X-t)\\
f_{-+}&=\frac14(1-\ang Z+\ang X-t)\\
f_{--}&=\frac14(1-\ang Z-\ang X+t),
\end{align*}
where $t$ is arbitrary real number.

Feynman considered the solution where $t$ is the expectation \ang{Y}
in the same state \ket{\psi} of the third Pauli operator
$Y = \begin{sm} 0&-i\\i&0 \end{sm}$.

A question arises whether there is a nonnegative grounding for \cS. The answer is affirmative.
For each state \ket\psi, there is a
choice of $t$ that makes all four components of $f$ nonnegative.

Indeed, write down the four inequalities $f_{\pm\pm}\geq0$ using the
formulas above for these $f_{\pm\pm}$'s.  Solve each one for $t$. You
find two lower bounds on $t$, namely
\begin{align*}
-1-\ang Z-\ang X &\text{ (from $f_{++}\ge0$)}\\
-1+\ang Z+\ang X &\text{ (from  $f_{--}\ge0$)},
\end{align*}
and two upper bounds, namely
\begin{align*}
1+\ang Z-\ang X &\text{ (from $f_{+-}\ge0$)}\\
1-\ang Z+\ang X &\text{ (from $f_{-+}\ge0$)}.
\end{align*}
An appropriate $t$ exists if and only if both of the lower bounds are
less than or equal to both of the upper bounds. That gives four
inequalities, which simplify to $-1\leq\ang Z\leq 1$ and $-1\leq\ang
X\leq 1$.  But these are always satisfied, because the eigenvalues of
$Z$ and $X$ are $\pm1$.
Thus, if
\begin{align*}
m &= \max\left(-1-\ang Z-\ang X, -1+\ang Z+\ang X\right),\\
M &= \min\left(1+\ang Z-\ang X, 1-\ang Z+\ang X\right),
\end{align*}
then $m\le M$ and the general solution above yields a nonnegative
grounding for the observation space \cS\ if and only if
$m\le t\le M$.

\subsection{Three-test version of Feynman's experiment}
\label{sub:f2}

Again, there is a source that repeatedly produces single qubits, all
in the same state \ket{\psi}, but this time around there are three
tests --- $X$, $Y$ and $Z$ --- which can be performed on each of these
qubits. Test $X$ (resp.,\ $Y$ or $Z$) is measuring the Pauli operator
$X$ (resp.,\ $Y$ or $Z$) in state \ket{\psi}.

This time around, the sample points of our observation space $\cS$ are
represented by words of length three in the alphabet $\{-,+\}$ where
the signs in each triple refer to $X$, $Y$ and $Z$
respectively. Accordingly, the sample space
\begin{equation*}
\Omega = \{---, --+, -+-, -++, +--, +-+, ++-, +++\}.
\end{equation*}
Enumerating these triples in the given lexicographic order, we obtain
an alternative representation
$ \Omega = \{0, 1, 2, 3, 4, 5, 6, 7\} $.

The domain of the probability distribution $\P_X$, corresponding to
test $X$, is the Boolean algebra generated by two events
$-**$ and $+**$ where $*$ stands for either $+$ or $-$,
and similarly for the domains of $\P_Y$ and $P_Z$.

{\tt Grounding problem.} Let $\P$ be an alleged signed ground
probability distribution for \cS.
Abbreviating $\P(---), \P(--+), \dots, \P(+++)$ to $f_0, f_1, \dots,
f_7$ respectively, and abbreviating \ang{X}, \ang{Y}, \ang{Z} to
$x,y,z$ respectively, we have
\begin{align}
f_0 + f_1 + f_2 + f_3 &= \P_X(-**)=(1-x)/2\\
f_4 + f_5 + f_6 + f_7 &= \P_X(+**)=(1+x)/2\\
f_0 + f_1 + f_4 + f_5 &= \P_Y(*-*)=(1-y)/2\\
f_2 + f_3 + f_6 + f_7 &= \P_Y(*+*)=(1+y)/2\\
f_0 + f_2 + f_4 + f_6 &= \P_Z(**-)=(1-z)/2\\
f_1 + f_3 + f_5 + f_7 &= \P_Z(**+)=(1+z)/2
\end{align}

It is easy to see that this system of six equations is equivalent to
the following system of four equations obtained by removing equations
(1), (3) and (5) and adding a new equation (0).
\begin{align}
f_0 + f_1 + f_2 + f_3 + f_4 + f_5 + f_6 + f_7 &=1\tag{0}\\
f_4 + f_5 + f_6 + f_7 &= (1+x)/2 \tag{2}\\
f_2 + f_3 + f_6 + f_7 &= (1+y)/2 \tag{4}\\
f_1 + f_3 + f_5 + f_7 &= (1+z)/2 \tag{6}
\end{align}
Indeed (0) follows from (1) and (2), while (2), (4) and (6) imply (1),
(3) and (5) respectively in the presence of (0).

Viewing $f_3, f_5, f_6, f_7$ as parameters allows us to give a general
solution of the grounding problem:
\begin{align*}
f_1 &= (1+z)/2 - f_3 - f_5 - f_7\\
f_2 &= (1+y)/2 - f_3 - f_6 - f_7\\
f_4 &= (1+x)/2 - f_5 - f_6 - f_7\\
f_0 &= 1 - (1+x)/2 - (1+y)/2 - (1+z)/2 + f_3 + f_5 + f_6 + 2f_7
\end{align*}

\begin{theorem}\label{thm:f}
For every state \ket{\psi}, there is a nonnegative grounding $\P$.
\end{theorem}

\begin{proof}
View the triple $(x,y,z)$ as a point in a three-dimensional real
vector space. Since $-1 \le x\le 1$, we have
$0 \le (1 + x)/2 \le 1$, and the same holds for $y$ and $z$.
Thus the point
\[ q = \big((1+x)/2,(1+y)/2,(1+z)/2\big) \]
lies in the cube $C$ with corners
\begin{align*}
&p_0=(0,0,0),\ p_1=(0,0,1),\ p_2=(0,1,0),\ p_3=(0,1,1),\\\
&p_4=(1,0,0),\ p_5=(1,0,1),\ p_6=(1,1,0),\ p_7=(1,1,1).
\end{align*}
(In fact, point $(x,y,z)$ lies on the unit sphere around $p_0$, so
that point $q$ lies on the sphere of radius $1/2$ inscribed in the
cube $C$, but this is not important for the current proof.)

For each $i = 0, 1, \dots, 7$, let $\vec p_i$ be the vector from $p_0$
to $p_i$, and let $\vec q$ be the vector from $p_0$ to $q$.
Since the cube $C$ is the convex closure of the points $p_0, \dots,
p_7$, there are nonnegative coefficients $f_0, \dots, f_7$ such that
\begin{equation}
\sum_{i=0}^7f_i = 1\quad\text{and}\quad
\sum_{i=0}^7 f_i \vec p_i = \vec q \tag{F}
\end{equation}
It is easy that (F) is equivalent to the system of equations (0),(2),(4),(6). Indeed, if we project the vector equation in (F) to the three axes, we get
\begin{align*}
0f_0 + 0f_1 + 0f_2 + 0f_3 +
1f_4 + 1f_5 + 1f_6 + 1f_7 &= (1+x)/2\\
0f_0 + 0f_1 + 1f_2 + 0f_3 +
0f_4 + 1f_5 + 1f_6 + 1f_7 &= (1+y)/2\\
0f_0 + 1f_1 + 0f_2 + 1f_3 +
0f_4 + 1f_5 + 0f_6 + 1f_7 &= (1+z)/2
\qedhere
\end{align*}
\end{proof}
The proof gives a four-dimensional polytope of nowhere negative
solutions of the grounding problem. The average of these solutions is
also a nowhere negative solution.

\subsubsection*{An explicit canonical nowhere negative solution of the
  grounding problem.}
Recall that each $i\in\{0,1,\dots,7\}$ represents a string $a_ib_ic_i$
in alphabet $\{-,+\}$. Define
\[f_i = \frac{1\pm x}2 \times \frac{1\pm y}2 \times \frac{1\pm z}2\]
where the signs of $x,y$ and $z$ are $a_i, b_i$ and $c_i$
respectively.

\begin{claim}
Coefficients $f_i$ satisfy equations (1)--(6).
\end{claim}

\begin{proof}
By symmetry, it suffices to prove that equations (1) and (2) are
satisfied. We prove (1):
\begin{align*}
& \left(\frac{1-x}2  \frac{1-y}2  \frac{1-z}2
+ \frac{1-x}2  \frac{1-y}2  \frac{1+z}2\right)
+ \left(\frac{1-x}2  \frac{1+y}2  \frac{1-z}2
+ \frac{1-x}2  \frac{1+y}2  \frac{1+z}2\right)\\
&= \left(\frac{1-x}2  \frac{1-y}2\right)
+ \left(\frac{1-x}2  \frac{1+y}2\right) = \frac{1-x}2
\end{align*}
The intended proof of (2) is similar; just replace $\frac{1-x}2$ with
$\frac{1+x}2$ above.
\end{proof}

\subsection{Schneider's experiment}
\label{sub:schneider}

We formalize the experiment described in David Schneider's blog post \cite{Schneider}, model it by means of an observation space, and then analyze the observation space. Schneider's experiment is a version of the Einstein-Podolsky-Rosen-Bohm thought experiment \cite{EPR,Bohm}.
Figure~\ref{fig:1}, together with the caption, is from article
\cite{Aspect}.

\begin{figure}[H]
\includegraphics[scale=1,trim=.6in 9.5in 0 0.3in,clip]{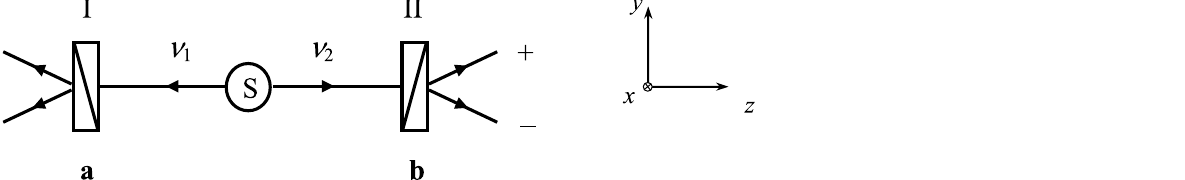}
\caption{The two photons, emitted in state \ket{\psi} are analyzed by
  linear polarizers in orientations $\bf a$ and $\bf b$. One can
  measure the probabilities of single or joint detections in the
  output channels of the polarizers.}
\label{fig:1}
\end{figure}

\noindent
In the figure, the two photons are moving along the $z$ axis. The pure
state $\ket\psi = \big(\ket{00} + \ket{11}\big)/\sqrt2$. The
computational basis vectors \ket0 and \ket1 correspond to polarization
in the directions of the $x$ and $y$ axes respectively%
\footnote{For polarization states of photons, the standard basis vectors correspond to vertical and horizontal polarization, in contrast to spin-$\frac12$ particles, whose standard basis vectors correspond to spin up and spin down.}

\begin{proposition}[\cite{PeresB}, \S6.2]\label{pro:peres}
If $\theta$ is the angle $\angle(\alpha,\beta)$ between the
orientations $\alpha$ and $\beta$ of the two polarizers in Figure~1,
then the outcomes $(+1,+1)$ and $(-1,-1)$, in which the two
measurements give us the same result, have probability
$\frac12\cos^2\theta$ each, and the outcomes $(+1,-1)$ and $(-1,+1)$
have probability $\frac12\sin^2\theta$ each.
\end{proposition}
\noindent
In particular, if $\alpha=\beta$ then the probability of getting the
same result is 1.

Our experiment involves three orientations $A, B, C$ with angles
\[ \angle(A, B) = \pi/4,\quad \angle(B, C) = \pi/8,\quad \angle(A, C)
  = 3\pi/8. \]
Let Alice and Bob manage the left and the right polarizers respectively.
To perform one run of the experiment, Alice gives her polarizer one of
the three orientations, and Bob gives his polarizer one of the three
orientations.

A priori this gives rise to nine tests: $(A,A), (A,B), \dots, (C,C)$.
But the three tests where Alice and Bob choose the same orientation
are not very interesting. We are going to ignore the runs of the
experiment with any of the three uninteresting tests. The remaining
six tests split into three groups of dual tests: $(A,B)$ with $(B,A)$,
$(B,C)$ with $(C,B)$, and $(A,C)$ with $(C,A)$.

Further, view every pair, $(\alpha,\beta)$ and $(\beta,\alpha)$, of
dual tests as insignificantly different variants of one symmetric test
where one party chooses orientation $\alpha$ and the other party
chooses orientation $\beta$. The symmetric test will be denoted simply
$\alpha\beta$ where the letters are in the lexicographical order. (So
the three symmetric tests are $AB,\ BC$ and $AC$.) The outcome
$(-1,+1)$ of the symmetric test $\alpha\beta$ means that the party
that chose $\alpha$ observed $-1$ while the party that chose $\beta$
observed $+1$. Similarly, the outcome $(+1,-1)$ means that the party
that chose $\alpha$ observed $+1$ and the party that chose $\beta$
observed $-1$.

There is a natural observation space \cS\ that models the
experiment. A sample point is a three-letter word $abc$ in
in the alphabet $\{-,+\}$. The intention is that $a=-$ (resp.\ $a=+$)
if a party choosing orientation $A$ would observe the result $-1$
(resp.\ $+1$). Of course, the same applies also to  $b$ and $c$.
The sample space is
\[\Omega = \{---,\ --+,\ -+-,\ -++,\ +--,\ +-+,\ ++-,\ +++\}. \]
Enumerating these triples in the given lexicographical order, we
obtain an alternative representation
$ \Omega = \{0, 1, 2, 3, 4, 5, 6, 7\} $.

Test $AB$ gives rise to a probability distribution $\P_{AB}$ whose
domain is the Boolean algebra of events generated by the four events
\begin{align*}
 --* &= \{---,\:--+\},\quad -+* = \{-+-,\:-++\},\\
 +-* &= \{+--,\:+-+\},\quad ++* = \{++-,\:+++\}
\end{align*}
determined by the values of $a$ and $b$. Similarly for the other two tests.
According to Proposition~\ref{pro:peres},
\begin{align*}\label{AB}
\P_{AB}(--*) &= \P_{AB}\{0,1\} = (1/2) \cos^2(\pi/4) = 1/4,\\
\P_{AB}(-+*) &= \P_{AB}\{2,3\} = (1/2) \sin^2(\pi/4) = 1/4,\\
\P_{AB}(+-*) &= \P_{AB}\{4,5\} = (1/2) \sin^2(\pi/4) = 1/4,\\
\P_{AB}(++*) &= \P_{AB}\{6,7\} = (1/2) \cos^2(\pi/4) = 1/4,\\[3pt]
\P_{BC}(*--) &= \P_{BC}\{0,4\} = (1/2) \cos^2(\pi/8)
= (1/4) + (\sqrt2/8),\\
\P_{BC}(*-+) &= \P_{BC}\{1,5\} = (1/2) \sin^2(\pi/8)
= (1/4) - (\sqrt2/8), \\
\P_{BC}(*+-) &= \P_{BC}\{2,6\} = (1/2) \sin^2(\pi/8)
= (1/4) - (\sqrt2/8),\\
\P_{BC}(*++) &= \P_{BC}\{3,7\} = (1/2) \cos^2(\pi/8)
= (1/4) + (\sqrt2/8),\\[3pt]
\P_{AC}(-*-) &= \P_{AC}\{0,2\} = (1/2) \cos^2(3\pi/8)
= (1/4) - (\sqrt2/8),\\
\P_{AC}(-*+) &= \P_{AC}\{1,3\} = (1/2) \sin^2(3\pi/8)
= (1/4) + (\sqrt2/8), \\
\P_{AC}(+*-) &= \P_{AC}\{4,6\} = (1/2) \sin^2(3\pi/8)
= (1/4) + (\sqrt2/8)\\
\P_{AC}(+*+) &= \P_{AC}\{5,7\} = (1/2) \cos^2(3\pi/8)
= (1/4) - (\sqrt2/8).
\end{align*}

The coherence requirement is easy to check. $\Dom{\P_{AB}} \cap
\Dom{\P_{AC}}$ is the Boolean algebra generated by two complementary
atoms $-**,\ +**$ determined by the values $a$. These atoms should
have probability $1/2$ with respect to both, $\P_{AB}$ and $\P_{AC}$,
and they do:
\begin{align*}
\P_{AB}(-**) &= \P_{AB}\{0,1,2,3\} = \P_{AC}\{0,1,2,3\} =
\P_{AC}(-**) = 1/2,\\
\P_{AB}(+**) &= \P_{AB}\{4,5,6,7\} = \P_{AC}\{4,5,6,7\} =
\P_{AC}(+**) = 1/2,
\end{align*}
The cases of $\Dom{\P_{AB}} \cap \Dom{\P_{BC}}$ and $\Dom{\P_{AC}} \cap \Dom{\P_{BC}}$ are similar.

\begin{proposition}
The groundings for observation space \cS\ form a family with
one real parameter, and no grounding is nonnegative.
\end{proposition}

\begin{proof}
Suppose that $\P$ is a grounding for \cS, and let
$t=\P(0)$. We have
\begin{align*}
\P(1) &= 1/4 - t
 &&\text{because }\P_{AB}\{0,1\} = 1/4\\
\P(2) &= 1/4 - (\sqrt2/8) - t
 &&\text{because }\P_{AC}\{0,2\} = (1/4) - (\sqrt2/8)\\
\P(3) &= (\sqrt2/8) + t
 &&\text{because }\P_{AB}\{2,3\} = 1/4\\
\P(4) &= 1/4 + (\sqrt2/8) - t
 &&\text{because }\P_{BC}\{0,4\} = (1/4) + (\sqrt2/8) \\
\P(5) &= -(\sqrt2/8) + t
 &&\text{because }\P_{AB}\{4,5\} = 1/4\\
\P(6) &= t
 &&\text{because }\P_{BC}\{2,6\} = (1/4) - (\sqrt2/8)\\
\P(7) &= 1/4 - t
 &&\text{because }\P_{AB}\{6,7\} = 1/4
\end{align*}
It is easy to check that, for any value of $t$, $\P$ is consistent
with $\P_{AB}$, $\P_{BC}$ and $\P_{AC}$. For no value of $t$ is $\P$ nonnegative, because
\[ \P\{2,5\} = (1/4) - (\sqrt2)/4 < 0. \qedhere \]
\end{proof}

\subsection{Hardy's experiment}
\label{sub:hardy}

Lucien Hardy described an interesting approach to non-locality in
\cite{Hardy}, and David Mermin elaborated it in \cite{Mermin}. The
following thought experiment is based on an example in \cite{Mermin}.

Two one-qubit particles emerge from a common source heading in
opposite directions toward distant detectors. Alice and Bob manage the
left and right detector respectively.
Aside from the passage of the particles from the source to the
detectors, there are no connections between Alice, Bob, and the
source.
Figure~\ref{fig:hardy} is from \cite{Mermin}.

\begin{figure}[H]
\includegraphics[scale=0.5,trim=.85in 8.6in 0in 1.3in,clip]{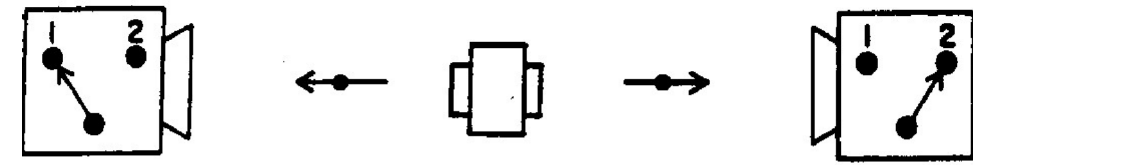}
\caption{}
\label{fig:hardy}
\end{figure}

\noindent
Ahead of each run of the experiment, Alice and Bob set their detectors
arbitrarily to one of two modes, indicated by ``1'' and ``2'' in the
figure.
When a particle arrives at a detector, the detector performs a
measurement and exhibits the result. In mode~1, Pauli observable $Z$
is measured,
and in mode~2, Pauli observable $X$ is measured.
Initially, the two particles are in the entangled state
\begin{equation}\label{zz}
 \ket\psi = \frac1{\sqrt3}\Big(\ket{01}+\ket{10}-\ket{00}\Big).
\end{equation}
Thus we have a four-test experiment $\E = \big(\ket\psi,\{Z\ox Z, Z\ox X, X\ox Z, X\ox X\}\big)$, performed jointly by Alice and Bob. The corresponding tests will be denoted $ZZ, ZX, XZ$ and $XX$.

The initial state is given to us in the computational basis
$\{\ket{00}, \ket{01}, \ket{10}, \ket{11}\}$. It will be convenient to
express it in three additional bases. Recall that $\ket+ =
(\ket0+\ket1)/\sqrt2$ and $\ket- = (\ket0 - \ket1) /\sqrt2$, and
therefore $\ket0 = (\ket+ + \ket-)/\sqrt2$ and $\ket1 = (\ket+ -
\ket-)/\sqrt2$.
We have
\[
\sqrt3\ket\psi = \ket0(\ket1-\ket0) + \ket{10}
 = \ket0(\ket1-\ket0) + \ket1(\ket+ + \ket-)/\sqrt2,
\]
and so, in basis $\{\ket{0+}, \ket{0-}, \ket{1+}, \ket{1-}\}$,
\begin{equation}\label{zx}
\ket\psi = -\frac{\sqrt2}{\sqrt3}\ket{0-} + \frac1{\sqrt6}\ket{1+} +
\frac1{\sqrt6}\ket{1-}.
\end{equation}
By the symmetry of \eqref{zz},
\begin{equation}\label{xz}
\ket\psi = -\frac{\sqrt2}{\sqrt3}\ket{-0}
+ \frac1{\sqrt6}\ket{+1}
+ \frac1{\sqrt6}\ket{-1}.
\end{equation}
Finally, in basis $\{\ket{++}, \ket{+-}, \ket{-+}, \ket{--}\}$, we have
\begin{equation}\label{xx}
\ket\psi
= \frac1{2\sqrt3}\ket{++} - \frac1{2\sqrt3}\ket{+-} -
  \frac1{2\sqrt3}\ket{-+} - \frac{\sqrt3}2\ket{--}.
\end{equation}

Next we describe an observation space \cS, starting with the sample
space $\Omega$.
In each run of the experiment, Alice chooses $Z$ or $X$, and so does Bob.
When either of them measures one of the observables ($Z$ or $X$), the result of the measurement is either $+1$ or $-1$, and (it will be useful to assume that) the detector exhibits $0$ or $1$ respectively.

The sample space $\Omega$ has 16 \emph{a priori} possible sample
points represented by binary strings
\[
Z_A X_A Z_B X_B
\]
where  $Z_A,Z_B$ are the symbols exhibited after measuring $Z$ by
Alice or Bob respectively, and $X_A,X_B$ are the symbols exhibited
after measuring $X$ by Alice or Bob respectively.

The four tests give rise to four probability distributions $\P_{ZZ},\
\P_{ZX},\ \P_{XZ}$ and $\P_{XX}$. The domain of each of these
distributions consists of the events that can be detected by the
corresponding test. Each of these four domains is a Boolean algebra of
sample points generated by four atoms. The probabilities are assigned
to the atoms in accordance with \eqref{zz}--\eqref{xx} respectively.

\begin{align}
&\text{\small Distributions}
&&\qquad\text{\small Atoms}
&&\quad\text{\small Probabilities}\phantom{mmmmmmmmm}
\nonumber\\
&\P_{ZZ}  &&\hspace{-10pt} 0*0*,\ 0*1*,\ 1*0*,\ 1*1*
          && 1/3,\  1/3,\  1/3,\  0\label{zz2} \\
&\P_{ZX}  &&\hspace{-10pt} 0**0,\ 0**1,\ 1**0,\ 1**1
          && 0,\ 2/3,\ 1/6,\ 1/6\label{zx2} \\
&\P_{XZ}  &&\hspace{-10pt} *00*,\ *01*,\ *10*,\ *11*
          && 0,\ 1/6,\ 2/3,\ 1/6\label{xz2} \\
&\P_{XX}  &&\hspace{-10pt} *0*0,\ *0*1,\ *1*0,\ *1*1
            \phantom{m}
          && 1/12,\ 1/12,\ 1/12,\ 3/4\label{xx2}
\end{align}
where $*$ means either of the binary digits.

To check the coherence requirement, we examine all six pairwise
intersections of our four domains.
$\Dom{\P_{ZZ}} \cap \Dom{\P_{XX}} = \Dom{\P_{ZX}} \cap \Dom{\P_{XZ}} =
\{ \emptyset, \Omega \}$.

$\Dom{\P_{ZZ}} \cap \Dom{\P_{ZX}}$ is the Boolean algebra generated by
two complementary atoms $0***$ and $1***$. It suffices to check that
one of the two atoms has the same probability in both
distributions. Indeed, we have
$\P_{ZZ}(0***) = 2/3 = \P_{ZX}(0***)$.
The case of $\Dom{\P_{ZZ}} \cap \Dom{\P_{XZ}}$ is similar. The
complementary atoms are $**0*$ and $**1*$, and $\P_{ZZ}(**0*) = 2/3 =
\P_{XZ}(**0*)$.

$\Dom{\P_{ZX}} \cap \Dom{\P_{XX}}$ is the Boolean algebra generated by
two complementary atoms $***0$ and $***1$. We have $\P_{ZX}(***0) =
1/6 = \P_{XX}(***0)$. The case of $\Dom{\P_{XZ}} \cap \Dom{\P_{XX}}$
is similar. The complementary atoms are $*0**$ and $*1**$, and
$\P_{XZ}(*0**) = 1/6 = \P_{XX}(*0**)$.
The coherence requirement is satisfied.

This completes the definition of observation space \cS. By \eqref{zz}, the initial state is symmetric with respect to the two qubits, and therefore the sides of Alice and Bob play symmetric roles.
The transformation
\begin{equation}\label{aut}
 Z_AX_A Z_BX_B \longrightarrow  Z_BX_B Z_AX_A
\end{equation}
is an automorphism of \cS. We will take advantage of that symmetry later.

By construction, \cS\ models \E.  The maps $\mu_{ZZ}, \mu_{ZX}, \mu_{XZ}, \mu_{XX}$ are defined in the obvious way. For example, the events $\mu_{ZZ}(E_{ab}(Z\ox Z))$ are the atoms $a*b*$ of $\P_{ZZ}$ with probabilities given by \eqref{zz2}; here $a,b\in\{0,1\}$.

Finally, we address the grounding problem for the observation space in question. There are no nonnegative groundings.
Indeed suppose toward a contradiction that $\P$ is such a distribution.
Observe that, thanks to nonnegativity, if an event has probability
zero, then so does each point in it.
By \eqref{zz2}--\eqref{xz2}, $\P(\omega)=0$ for any sample point
$\omega$ in $1*1*$ or $0**0$ or $*00*$; call these points $\omega$
idle. The remaining points are
\[
0011,\ 0101,\ 0111,\ 1100,\ 1101.
\]
By \eqref{zx2}, $\P(1100) = 1/6$ because $1100$ is the only non-idle
point in set $1**0$. On the other hand,
by \eqref{xx2}, $\P(1100)=1/12$ because $1100$ is the only non-idle
point in $*1*0$. This gives the desired contradiction.

Now let's consider the general case where negative probabilities are
allowed. The grounding problem reduces to solving a system of linear
equations in 16 variables. Let $\P$ be an alleged grounding. To avoid fractions, we deal with variables
\[ v_0 = 12\P(0000),\ v_1 = 12\P(0001),\ \dots,\
   v_{14} = 12\P(1110),\ v_{15} = 12\P(1111). \]
Notice that each subscript of $v$ is the number represented in binary
by the corresponding argument of $\P$.
Each of the 16 atoms in \eqref{zz2}--\eqref{xx2} gives rise to a
linear equation in four variables:
\begin{align}\label{sys1}
0*0*&\ \textrm{ generates}\phantom{mmmmmmm}
  &&v_0 + v_1 + v_4 + v_5 &&= 4\nonumber\\
0*1*&\ \textrm{ generates}
  &&v_2 + v_3 + v_6 + v_7 &&= 4\nonumber\\
1*0*&\ \textrm{ generates}
  &&v_8 + v_9 + v_{12} + v_{13} &&= 4\nonumber\\
1*1*&\ \textrm{ generates}
  &&v_{10} + v_{11} + v_{14} + v_{15} &&= 0\nonumber\\
0**0&\ \textrm{ generates}
  &&v_{0} + v_{2} + v_{4} + v_{6} &&= 0\nonumber\\
0**1&\ \textrm{ generates}
  &&v_{1} + v_{3} + v_{5} + v_{7} &&= 8\nonumber\\
1**0&\ \textrm{ generates}
  &&v_{8} + v_{10} + v_{12} + v_{14} &&= 2\nonumber\\
1**1&\ \textrm{ generates}
  &&v_{9} + v_{11} + v_{13} + v_{15} &&=2\\
*00*&\ \textrm{ generates}
  &&v_{0} + v_{1} + v_{8} + v_{9} &&= 0\nonumber\\
*01*&\ \textrm{ generates}
  &&v_{2} + v_{3} + v_{10} + v_{11} &&= 2\nonumber\\
*10*&\ \textrm{ generates}
  &&v_{4} + v_{5} + v_{12} + v_{13} &&= 8\nonumber\\
*11*&\ \textrm{ generates}
  &&v_{6} + v_{7} + v_{14} + v_{15} &&= 2\nonumber\\
*0*0&\ \textrm{ generates}
  &&v_0 + v_2 + v_{8} + v_{10} &&= 1\nonumber\\
*0*1&\ \textrm{ generates}
  &&v_1 + v_3 + v_{9} + v_{11} &&= 1\nonumber\\
*1*0&\ \textrm{ generates}
  &&v_4 + v_6 + v_{12} + v_{14} &&= 1\nonumber\\
*1*1&\ \textrm{ generates}
  &&v_5 + v_7 + v_{13} + v_{15} &&= 9\nonumber
\end{align}

Any solution of this system of 16 equations with 16 variables
determines a grounding for our observation space. More
interesting and relevant groundings are those which conform
to the  automorphism \eqref{aut} which has four fixed points $0000,\
0101,\ 1010,\ 1111$ and pairs the remaining twelve points as follows:
\begin{align*}
0001 &\leftrightarrow 0100,\ 0010 \leftrightarrow 1000,\ 0011
       \leftrightarrow 1100,\\
0110 &\leftrightarrow 1001,\ 0111 \leftrightarrow 1101,\ 1011
       \leftrightarrow 1110.
\end{align*}
The conformant groundings give rise to \emph{symmetric
  solutions} of system \eqref{sys1}, subject to the following
constraints:
\begin{equation}\label{dual}
 v_1 = v_4,\ v_2 = v_8,\ v_3 = v_{12},\ v_6 = v_9,\ v_7 = v_{13},\  v_{11} = v_{14}.
\end{equation}
Notice that any solution $\vec v$ of \eqref{sys1} leads to a symmetric
solution. Indeed, let $\vec w$ be obtained from $\vec v$ by swapping
the values of the variables $v_i, v_j$ whenever the equality $v_i =
v_j$ occurs in \eqref{dual}. Due to the automorphism \eqref{aut}, if
$\vec v$ is a solution for system \eqref{sys1} then so is $\vec
w$. But the solutions for \eqref{sys1} form an affine space, and so
the average $(\vec v + \vec w)/2$ is a symmetric solution for
\eqref{sys1}.

Here we restrict attention to symmetric solutions of \eqref{sys1}. The constraints \eqref{dual} simplify the problem considerably. 0**0 and *00* generate the same equation in \eqref{sys1}, and so do pairs (0**1,*10*), (1**0,*01*), and (1**1,*11*). Accordingly, we have 12 equations with 10 variables
\[\phantom{mmn}  \begin{pmatrix} v_0,& v_1,& v_2,& v_3,& v_5,& v_6,& v_7,& v_{10},& v_{11},& v_{15}\end{pmatrix}. \]
We spare the reader the details of solving this system of linear
equations. The general solution can be given in the following form
\begin{align*}
\begin{matrix}
\qquad (-3, & -1, & 2, & 0, & 9, & 2, & 0, & 0, & 0, & 0) \\
+\ a\cdot (\phantom{-}2,& -1,& -1,& 1,& 0,& 0,& 0,& 0,& 0,& 0) \\
+\ b\cdot (\phantom{-}0,& 1,& 0,& 0,& -2,& -1,& 1,& 0,& 0,& 0) \\
+\ c\cdot (\phantom{-}0,& 0,& 1,& 0,& 0,& -1,& 0,& -2,& 1,& 0) \\
+\ d\cdot (-1,& 1,& 1,& 0,& -1,& -1,& 0,& -1,& 0,& 1)
\end{matrix}
\end{align*}
where $a,b,c,d$ are real parameters.

\section{Contextuality and nonnegative groundings}
\label{sec:context}
%

For simplicity, as in \S\ref{sec:obs} and \S\ref{sec:finite}, we consider finite-dimensional Hilbert spaces, so that the spectra of observables are pure point spectra.
Recall that measuring several commuting observables produces eigenvalues for a common eigenvector of those observables, in short simultaneous eigenvalues.

In a hidden-variable theory, measurement outcomes are determined by the values of the hidden variables and thus exist prior to the measurement.
They are merely revealed by the measurement. It does not matter whether you measure an observable $A$ all by itself (context 1) or together with another compatible observable $B$ (context 2). In that sense, hidden-variable theories are non-contextual.
By the Kochen-Specker theorem \cite{KS}, non-contextual theories cannot reproduce all predictions of quantum mechanics in Hilbert spaces of dimension $\ge3$.

Notice that an observation space \cS\ that models a multi-test experiment $\E = \big( \ket\psi, \O \big)$  constitutes a hidden-variable model of \E\ with one hidden variable whose possible values are the sample points.
The partition requirement in the definition of modeling ensures that each sample point $\omega$ assigns a unique eigenvalue to any observable $A\in\O$:
if $\omega$ is in $\mu_A E_r(A)$ then it assigns to $A$ the  value $r$.
Accordingly, $\mu_A E_r(A)$ is the event that measuring $A$ in the state \ket\psi\ produces eigenvalue $r$.

The correctness requirement assures that the probabilities $\P_A(\mu_A E_r(A))$ match the predictions of quantum theory.

\begin{definition}\label{def:context}
We say that a set \O\ of observables on \H\ is \emph{contextual} if, for every function $f$ assigning to each $A\in\O$ one of its eigenvalues, there are commuting observables $A_1,\dots,A_k\in \O$ such that the values $f(A_1),\dots, f(A_k)$ are not simultaneous eigenvalues of $A_1,\dots,A_k$. \qef
\end{definition}

This definition is stricter than some in the literature, but the standard examples of contextuality provide sets \O\ that satisfy our definition of contextuality.
Kochen and Specker proved that, in the three-dimensional Hilbert space, there is a finite contextual set of rank-1 projections \cite{KS}.
A simpler example is described in \cite{Cabello} and nicely illustrated in the Wikipedia article \cite{W-KS}.
Other simple contextuality examples are described in \cite{Mermin} and \cite{PeresA}.

For instance, let us indicate how the example in \cite{Cabello} satisfies our definition. That example involves 18 vectors in $\C^4$, forming 9 bases of four vectors each, with each of the 18 vectors in exactly two of the bases.
Let \O\ consist of the 18 projections to the 18 vectors plus an observable for each base, having the four vectors of the base as eigenvectors for four distinct eigenvalues.
Let $f$ be any function as in Definition~\ref{def:context}. For each of the 9 bases, $f$ picks out one eigenvalue for the associated observable and thus picks out one of the basis vectors. Since 9 is odd, one of the 18 vectors, call it $v$, is chosen from exactly one of the two bases where it occurs. Let $A$ be the projection to $v$.

If $f(A) = 0$, let $B$ be the observable corresponding to the base where $f$ chose $v$. Then $f(A), f(B)$ are not simultaneous eigenvalues of $A,B$. Otherwise, $f(A) = 1$. In this case, let $B$ be the observable corresponding to the other base containing $v$, i.e.\  where $f$ could choose $v$ but didn't. Again, $f(A), f(B)$ are not simultaneous eigenvalues of $A,B$. In either case, we get a witness for the contextuality of \O.

Concerning the Peres-Mermin square \cite{Mermin,PeresA}, the usual contextuality argument shows that the square also satisfies our definition of contextuality with $k=3$.

Recall the modeling Definition~\ref{def:model}, the notation
introduced just before that definition, and the monotonicity
Definition~\ref{def:monotone}.

\begin{theorem}\label{thm:context}
Let \O\ be a contextual set of observables.
Then no multi-test experiment $\E = \left(\ket\psi, \O\right)$ can be monotonically modeled by an observation space.
\end{theorem}

\begin{proof}
Suppose, toward a contradiction, that $\cS = \left(\Omega, \ang{\P_A: A\in\O}\right)$ is an observation space that monotonically models \E\ via  maps $\mu_A: \ES(A)
\to 2^\Omega$.

Fix any sample point $\omega\in\Omega$. For each $A\in\O$ and each
eigenvalue $r$ of $A$, let $f(A) = r$ if $\omega \in
\mu_A(E_r(A))$. By the partition requirement, $f$ is well
defined.

Since \O\ is contextual, there are commuting $A_1,\dots,A_k$ in \O\
such that $f(A_1),\dots, f(A_k)$ are not simultaneous eigenvalues of $A_1,\dots,A_k$.
Lemma~\ref{lem:sm} provides an observable $B\in\O$ whose eigenspaces are intersections of eigenspaces of $A_1,\dots,A_k$.
In particular, letting $b=f(B)$, we have that $E_b(B)=E_{a_1}(A_1)\cap \cdots \cap E_{a_k}(A_k)$ for some eigenvalues $a_1, \dots, a_k$ of $A_1,\dots,A_k$.

Since $B$ commutes with $A_1,\dots,A_k$, monotonicity implies that $\omega$, being in $\mu_B(E_b(B))$, is
also in all $\mu_{A_i}(E_{a_i}(A_i))$.  That is, $f(A_i)=a_i$
for all $i$. But $(a_1,\dots, a_k)$ are simultaneous eigenvalues of $A_1,\dots,A_k$, witnessed by any eigenvector in $E_b(B)$, and this contradicts our choice of $A_1,\dots,A_k$.
\end{proof}

\begin{remark}
In the preceding proof, if observation space \cS\ admits a non-negative grounding \P, the sample point $\omega$ might have $\P(\omega)=0$.
In such a situation, the partition requirement of Definition~\ref{def:model} seems too strong, at least in the case of discrete distributions.
Why should an impossible outcome $\omega$ produce well defined values for observables?
Fortunately, in that case (of discrete distributions), we can modify the proof and achieve that the relevant sample point $\omega$ has $\P(\omega)>0$.

By definition, the domain of \P\ is generated by observable events.
Without loss of generality, we may assume the following.
\begin{enumerate}
\item Every two sample points are distinguished by some observable event.
\item There are no sample points $\omega$ with probability $\P(\omega) = 0$.
\end{enumerate}
To achieve (1), for every maximal subset $S$ of $\Omega$ of
cardinality $\ge2$ such that no two points of $S$ are distinguished by
the observable events, merge the points of $S$ into one point of
probability $\sum_{\omega\in
  S}\P(\omega)$. Then, to achieve (2), just discard sample points
$\omega$ of probability $\P(\omega) = 0$. It is obvious that the
reduced observation space still admits a nonnegative grounding and
monotonically models \E. As a result, each sample point $\omega$ has a positive probability $\P(\omega)$. \qef

\end{remark}

\section{Wigner's distribution as a ground distribution}
\label{sec:w}

This section can be read independently from all the previous sections. But, even though we don't mention observation spaces explicitly, the reader will recognize that we deal with groundings for a particular observation space.
Contrary to \S\S3--5, where we restricted attention to observables with pure point spectra, here we consider observables of a very different kind.

Heisenberg's uncertainty principle asserts a limit to the precision
with which position $x$ and momentum $p$ of a particle can be known
simultaneously. One can know the probability distributions of $x$ and
$p$ individually, yet the joint probability distribution with the
correct marginal distributions of $x$ and $p$ makes no physical
sense.
But maybe such a joint distribution makes mathematical sense.

In 1932, Eugene Wigner exhibited a joint signed distribution
\cite{Wigner} with the correct marginal distributions of $x$ and
$p$. Some of the values of Wigner's distribution are negative. ``But of course this must not hinder the use of it in calculations as an auxiliary function which obeys many relations we would expect from such a probability'' \cite[p.~751]{Wigner}.
Wigner's ideas indeed have been used in optical tomography; see
\cite{Smithey} for example.

According to \cite{Bertrand}, Wigner's distribution is the unique signed probability distribution that yields the correct marginal distributions for position and momentum and all their linear combinations. Here we rigorously prove that assertion%
\footnote{When we found a simple proof of the assertion, we attempted, in \cite{G224}, to popularly explain that proof. Here we replace all hand-waving in our explanation with rigorous mathematical arguments. Fortunately, the proof remains relatively simple.}.
For simplicity we work with one particle moving in one dimension, but everything we do in this section generalizes in a routine way to more (distinguishable non-relativistic) particles in more dimensions.

In this context we have to deal with the Hilbert space $L^2(\R)$ of
square integrable functions $f: \R \to\C$ where the inner product
\braket fg is given by the Lebesgue integral
\[ \int_{-\infty}^\infty f^*(x)g(x) dx. \]
  In the rest of this section, by default, the
integrals are from $-\infty$ to $\infty$.

A state of the particle is given by a unit vector $\ket\psi$  in $L^2(\R)$.
The position and momentum are given by Hermitian operators $X$ and $P$ where
\[
 (X\psi)(x)= x\cdot\psi(x)\quad\text{and}\quad
 (P\psi)(x)= -i\hbar\frac{d\psi}{dx}(x)
\]
and $\hbar$ is the (reduced) Planck constant.
For any real numbers $a,b$, not both zero, consider the Hermitian
operator (an observable) $Z=aX+bP$.

The unbounded operator $Z$ is not defined on the whole state space
$L^2(\R)$. But its domain clearly includes the space $C^\infty_c(\R)$
of smooth, compactly supported functions on $R$. For brevity, we will
say that a state \ket{\psi} in $L^2(\R)$ is \emph{nice} if the
function $\ket{\psi} \in C^\infty_c(\R)$. The restriction of the
operator $Z$ to $C^\infty_c(\R)$ is essentially self-adjoint
\cite[Proposition~9.40]{Hall} which is implicitly used below.

Fix a nice state \ket{\psi} and consider an observation space
\begin{equation}\label{os}
\left(\R\times\R, \ang{\P_{a,b}:\ a,b\in\R
\textrm{ and }a^2+b^2\ne0} \right)
\end{equation}
where the domain of any $\P_{a,b}$ consists of sets
$ [ax+bp\in E] = \big\{(x,p): ax+bp\in E\big\}, $
wherein $E$ ranges over Borel subsets of $\R$, and $\P_{a,b}[ax+bp\in
E]$ is the quantum mechanical probability $\P_\psi[aX+bP\in E]$ that
the  measurement of $aX+bP$ in \ket{\psi} produces a value in $E$.

To check the consistency requirement, let $e = [ax+bp \in E] = [a'x +
b'p\in E'] \ne\emptyset$.
Pick any point $(x_0,p_0)\in e$. Since $e = [ax+bp \in E]$, the value
$c = ax_0 + bp_0$ belongs to $E$, and $e$ includes the line $L$ given
by equation $ax + bp = c$.
Pick any point $(x',p')\in L$. Since $L\subseteq e = [a'x + b'p\in
E']$, the value $c' = a'x' + b'p'$ belongs to $E'$, and $e$ includes
the line $L'$ given by equation $a'x + b'p = c'$
If $L' \ne L$, then $e$ includes all lines parallel to $L'$, because
$(x',p')$ could be any point on $L$. Thus $e = \R^2$ and $1 =
\P_{a,b}(e) = \P_{a',b'}(e)$.
If $L' = L$, then $a' = ar,\ b' = br,\ c' = cr$ for some $r$, and $E'
= rE$, and $\P_{a,b}(e) = \P_\psi[aX+bP\in E] = \P_\psi[a'X+b'P\in E']
= \P_{a',b'}(e)$.

We restricted the range of the variable $E$ to Borel sets because it
is sufficient for our purposes and because the functional calculus
guarantees that, for Borel sets $E$, the probabilities
$\P_\psi[aX+bP\in E]$ are well defined \cite[\S6]{Hall}.

\begin{lemma}[\cite{Volberg}]\label{lem:sasha}
For any real $a,b$, not both zero, there is a probability density
function $g(z)$ for $Z = aX + bP$, so that for every Borel set $E$
\[ \P_\psi[Z\in E] = \int_E g(z) dz. \]
\end{lemma}

\begin{proof}
If $Z$ is $X$ or $P$, the claim is well known. Suppose $b\ne0$.
To simplify notation, assume without loss of generality that $b = 1$,
and let $c = -a/\hbar$ so that $Z = -c\hbar X + P$.
Let $U$ be the unitary operator $f\mapsto e^{icx^2/2}f$ over
$L^2(\R)$.
For any differentiable function $f(x)\in L^2(\R)$, we have
\[ (UPU^{-1})f = -i\hbar \cdot e^{icx^2/2}\frac d{dx} \left(
    e^{-icx^2/2} f \right)  = -c\hbar xf - i\hbar f' = -c\hbar Xf + Pf
  = Zf. \]
Since $Z$ is the unitary conjugate of $P$, there is a density function
for $Z$ in state \ket{\psi} because there is a density function for
$P$ in the nice state $U^{-1}\ket\psi$.
\end{proof}


Wigner's distribution is the signed probability distribution on $\R^2$
with the density function
\begin{equation} \label{wd}
w(x,p) = \frac1{2\pi} \int \psi^*(x+\frac{\beta\hbar}2)
\psi(x-\frac{\beta\hbar}2) e^{i\beta p}\,d\beta.
\end{equation}
All values of Wigner's density functions are real. Indeed, if
$F(x,p,\beta)$ is the integrand in \eqref{wd}, then
\[ F^*(x,p,\beta) =
  \psi(x+\frac{\beta\hbar}2)\psi^*(x-\frac{\beta\hbar}2)
e^{-i\beta p} = F(x,p,-\beta),\]
and therefore
\[ w^*(x,p) = \frac1{2\pi}\int_{-\infty}^\infty F^*(x,p,\beta)d\beta
 = \frac1{2\pi}\int_{-\infty}^\infty F(x,p,-\beta)d\beta = w(x,p).\]
On the other hand, the values of $w(x,p)$ may be negative.

\begin{theorem}\label{thm:w}%
\mbox{}\hspace{-5pt}\footnotemark\
In every nice state \ket{\psi}, Wigner's density function
is the unique signed density function on $\R^2$ that yields the
correct marginal density functions for all linear combinations
$Z=aX+bP$ where $a,b$ are real numbers, not both zero.
\end{theorem}

\footnotetext{A stronger version of the theorem (not requiring the state to be nice) with a simpler proof is in the article ``Wigner's quasidistribution and Dirac's kets'' at \url{https://arxiv.org/abs/2201.05911}}

\begin{corollary}
For every nice state \ket{\psi}, Wigner's distribution is the unique continuous grounding on the Borel sets in $\Omega$
for the observation space \eqref{os}.
\end{corollary}

In the rest of this section, we prove Theorem~\ref{thm:w}. Throughout
the proof, we use $(a,b)$ to denote pairs of real numbers not both
zero.

If $f(x,p)$ is a signed probability density function on $\R^2$, and $z
= ax+bp$, then the marginal density function $g(z)$ can be defined
thus:
\begin{equation}\label{marginal}
g(z) =
\begin{cases}
 \displaystyle
 \frac1{|b|} \int f(x, \frac1b (z-ax))\,dx
 &\mbox{if $b\ne0$}\\[10pt]
 \displaystyle
 \frac1{|a|} \int f(\frac1a (z-bp), p)\,dp
 &\mbox{if $a\ne0$}
\end{cases}
\end{equation}
Here's a justification in the case $b\ne0$. For any real $u\le v$, the
probability that $u\le z\le v$ should be
\[
\int_u^v g(z) dz= \iint_{u\le ax+bp\le v} f(x,p)\,dx\,dp.
\]
Since $p  = \frac1b(z-ax)$, we can change variables in the integral,
from $x,p$ to $x,z$. The absolute value of the Jacobian determinant of
this transformation is $1/|b|$, so we obtain
\begin{align*}
\int_u^v g(z) dz
&= \iint_{u\le ax+bp\le v}
   \frac1{|b|} f\big(x,\frac1b(z-ax)\big)\,dx\,dz \\
&= \int_u^v dz \int_{-\infty}^\infty
   \frac1{|b|} f\big(x,\frac1b(z-ax)\big)\,dx.
\end{align*}
Since the first and last expressions coincide for all $u\le v$, we
have
\[
  g(z) = \frac1{|b|} \int f\big(x,\frac1b(z-ax)\big)\,dx.
\]

\begin{lemma}\label{lem:j2m}
For any $a,b$ not both zero and $f$ as above, the following claims are equivalent.
\begin{enumerate}
\item $g(z)$ is the marginal probability density function of $z = ax + bp$.
\item $\displaystyle \hat g(\zeta) =
  \sqrt{2\pi}\cdot\hat f(a\zeta,b\zeta)$ where $\hat f$ and $\hat g$
  are (forward) Fourier transforms of $f$ and $g$ respectively.
\end{enumerate}
\end{lemma}

\begin{proof}
We assume $b \neq 0$; the case $a \neq 0$ is similar.
To prove (1)$\to$(2), suppose (1) and compare the forward Fourier
transforms of $g$ and $f$:
\begin{align*}
 \hat g(\zeta)&= \frac1{\sqrt{2\pi}}\int g(z)e^{-i\zeta z}\,dz\\
 &=\frac1{\sqrt{2\pi}}\iint f\big(x,\frac1b(z-ax)\big)
   e^{-i\zeta z}\frac1{|b|}\,dx\,dz\\
 &=\frac1{\sqrt{2\pi}}\iint f(x,p)e^{-i\zeta(ax+bp)}\,dx\,dp.
\end{align*}
Since $\displaystyle\hat f(\xi,\eta) =
   \frac1{2\pi}\,\iint f(x,p)e^{-i(\xi x+\eta p)}\,dx\,dp$,
we have
$\hat g(\zeta) = \sqrt{2\pi} \hat f(a\zeta,b\zeta)$.

To prove (2)$\to$(1), suppose (2) and use the implication
(1)$\to$(2). If $h$ is the marginal distribution of $z = ax + by$ then
$ \hat h(\zeta) = \sqrt{2\pi}\cdot\hat f(a\zeta,b\zeta) =
 \hat g(\zeta)$,
and therefore $g = h$.
\end{proof}

\begin{corollary}\label{cor:j2m}
For any real $\alpha,\beta$ not both zero,
$\displaystyle\hat f(\alpha,\beta) =
\frac1{\sqrt{2\pi}}\, \hat g(\zeta)$
where
$g(z)$ is the marginal distribution for the linear combination
$z=ax+bp$ such that $\alpha=a\zeta$, $\beta=b\zeta$ for some $\zeta$.
\end{corollary}

\begin{lemma}\label{lem:key}
For any nice state \ket{\psi} and any real numbers $\alpha,\beta$, not
both zero, we have
\[
 \bra\psi e^{-i (\alpha X + \beta P)}\ket\psi =
 e^{i\alpha\beta\hbar/2}
 \int\psi^*(y)e^{-i\alpha y}\psi(y-\beta\hbar)\,dy.
\]
\end{lemma}

\begin{proof}
We want to split the exponential on the left side of this equation
into a factor with $X$ times a factor
with $P$. This is not as easy as it might seem, because $X$ and $P$
don't commute.  We have, however, two pieces of good luck.  First,
there is Zassenhaus's formula \cite{Casas+}, which expresses the
  exponential of a sum of non-commuting quantities as a product of
  (infinitely) many exponentials, beginning with the two that one
  would expect from the commutative case, and continuing with
  exponentials of nested commutators:
\[
 e^{A+B}=e^Ae^Be^{-\frac12[A,B]}\cdots,
\]
where the ``$\cdots$'' refers to factors involving double and higher
commutators.

The second piece of good luck is that $[X,P]=i\hbar I$, where $I$ is
the identity operator.  (Below, we'll usually omit writing $I$
explicitly, so we'll regard this commutator as the scalar $i\hbar$.)
Since that commutes with everything, all the higher commutators in
Zassenhaus's formula vanish, so we can omit the ``$\cdots$'' from the
formula. We have
\[
 \bra\psi e^{-i\alpha X-i\beta P}\ket\psi\\
   = \bra\psi e^{-i\alpha X}e^{-i\beta P} e^{i\alpha\beta\hbar/2}\ket\psi.
\]
The last of the three exponential factors here arose from Zassenhaus's
formula as
\[
 -\frac12[-i\alpha X,-i\beta P] = \frac12\alpha\beta[X,P] =
 i\alpha\beta\hbar/2.
\]
That factor, being a scalar, can be pulled out of the bra-ket:
\[
 \bra\psi e^{-i (\alpha X + \beta P)}\ket\psi=
 e^{i\alpha\beta\hbar/2}
 \int \psi^*(y) e^{-i\alpha y} \psi(y-\beta\hbar)\,dy. \qedhere
\]
\end{proof}

Finally, we are ready to prove Theorem~\ref{thm:w}.
Assume that, in a given nice state \ket\psi, a signed probability
density function $f(x,p)$ yields correct marginal density functions
for all linear combinations of position and momentum.

Fix any real $a,b$ not both zero, and let $Z$ be the operator
$aX + bP$ and $z = ax + bp$ where $x,p$ are real variables. By
Lemma~\ref{lem:sasha}, there is a probability density function
$g_0(z)$ for $Z$. For any real $\zeta$, we have
\begin{align}\label{w0}
\frac1{\sqrt{2\pi}} \hat g_0(\zeta)
&=\frac1{2\pi}\int g_0(z)e^{-i\zeta z}\,dz \nonumber\\
 &=\frac1{2\pi}\bra\psi e^{-i\zeta Z}\ket\psi \\
 &=\frac1{2\pi}\bra\psi e^{-i\zeta (aX+bP)}\ket\psi.\nonumber
\end{align}
Let $g(z)$ be the marginal distribution of $f$ given by
formula~\eqref{marginal}. By the assumption, $g$ coincides with
$g_0$. By Corollary~\ref{cor:j2m}, we have

\begin{equation}\label{w1}
\hat f(\alpha,\beta)
= \frac1{\sqrt{2\pi}} \hat g(\zeta)
=\frac1{2\pi}\bra\psi e^{-i\zeta (aX+bP)}\ket\psi.
\end{equation}
By Lemma~\ref{lem:key},
\begin{equation}\label{w2}
\hat f(\alpha,\beta)
= \frac{e^{i\alpha\beta\hbar/2}} {2\pi}
   \int\psi^*(y)e^{-i\alpha y}\psi(y-\beta\hbar)\,dy.
\end{equation}
To get $f(x,p)$, apply the (two-dimensional) inverse Fourier
transform.
\[
 f(x,p)=\frac1{(2\pi)^2}\iiint
 \psi^*(y)e^{-i\alpha y}e^{i\alpha\beta\hbar/2}\psi(y-\beta\hbar)
 e^{i\alpha x}e^{i\beta p}\,dy\,d\alpha\,d\beta.
\]
Collecting the three exponentials that have $\alpha$ in the exponent,
and noting that $\alpha$ appears nowhere else in the integrand,
perform the integration over $\alpha$ and get a Dirac delta function:
\[
 \int e^{-i\alpha(y-\frac{\beta\hbar}2-x)}\,
 d\alpha=2\pi\delta(y-x-\frac{\beta\hbar}2).
\]
That makes the integration over $y$ trivial, and what remains is
\begin{equation}\label{w3}
f(x,p)=\frac1{2\pi}\int\psi^*(x+\frac{\beta\hbar}2)\psi(x-\frac{\beta\hbar}2)
e^{i\beta p}\,d\beta,
\end{equation}
which is Wigner's signed probability distribution $w(x,p)$.

To check that Wigner's signed probability distribution yields correct
marginal distributions, note that the derivation of \eqref{w3} from
\eqref{w1} is reversible. By \eqref{w1} and \eqref{w0}
\[
\frac1{\sqrt{2\pi}} \hat g(\zeta)
=\frac1{2\pi}\bra\psi e^{-i\zeta (aX+bP)}\ket\psi
= \frac1{\sqrt{2\pi}} \hat g_0(\zeta),
\]
so that $\hat g = \hat g_0$ and therefore $g = g_0$.
This completes the proof of Theorem~\ref{thm:w}.

\section*{Conclusion}

To address the question what negative probabilities are good for, we introduced observation spaces and demonstrated that all quantum multi-test experiments can be modeled by observation spaces and their non-negative groundings.

We also demonstrated that in specific examples modeling can be made more succinct and faithful. This comes with a price: in some cases negative probabilities are unavoidable in these groundings.
But we agree with Wigner that this fact ``must not hinder'' the use of negative probabilities in calculations \cite{Wigner}. In fact, Wigner's distribution, which usually has some negative values, is widely used in quantum optics.

On the other hand, we showed that a situation proposed by Feynman as an application of negative probabilities admits a succinct non-negative grounding. This remains true even in an extended version of Feynman's example.

We showed that succinct faithful modeling is not universally available.  In contextual situations like those in the Kochen-Specker theorem, such modeling is impossible.

Coming back to Wigner's distribution, we rigorously proved that this distribution is uniquely characterized by its marginal distributions of linear combinations of position and momentum.

\end{document}